\documentclass{article}
\usepackage[a4paper, margin=1in]{geometry} 
\usepackage{graphicx} 
\usepackage{booktabs}
\usepackage{multirow}
\usepackage{amsthm}
\usepackage{amsmath} 
\usepackage{amssymb}

\usepackage{xcolor}

\newtheorem{Theorem}{Theorem}[section]

\title{Local Clustering in Hypergraphs through Higher-Order Motifs }

\author{
Giuseppe F. Italiano\thanks{Funded by the Italian Ministry of University and Reseach under PRIN Project n. 2022TS4Y3N - EXPAND: scalable algorithms for EXPloratory Analyses of heterogeneous and dynamic Networked Data.}\\
LUISS University, Italy\\
\texttt{gitaliano@luiss.it}
\and
Athanasios L. Konstantinidis\\
University of Ioannina, Greece\\
\texttt{a.konstantinidis@uoi.gr}
\and
Anna Mpanti\thanks{Funded by the European Union - Next Generation EU, Mission 4 Component 2 CUP: I83C22000990001.}\\
LUISS University, Italy\\
\texttt{ampanti@luiss.it}
\and
Fariba Ranjbar\thanks{Supported by PRIN 2022 project "Next Generation Alogorithms for Constrained Graph Visualization", 2022MI9Z78,  MIUR, (Italian Ministry of Education and Research)}\\
LUISS University, Italy\\
\texttt{fariba.ranjbar@luiss.it}
}

\date{}

\begin{document}

\maketitle

\begin{abstract}

Hypergraphs provide a powerful framework for modeling complex systems and networks with higher-order interactions beyond simple pairwise relationships. However, graph-based clustering approaches, which focus primarily on pairwise relations, fail to represent higher-order interactions, often resulting in low-quality clustering outcomes. In this work, we introduce a novel approach for local clustering in hypergraphs based on higher-order motifs, small connected subgraphs in which nodes may be linked by interactions of any order, extending motif-based techniques previously applied to standard graphs. Our method exploits hypergraph-specific higher-order motifs to better characterize local structures and optimize motif conductance. We propose two alternative strategies for identifying local clusters around a seed hyperedge: a core-based method utilizing hypergraph core decomposition and a BFS-based method based on breadth-first exploration. We construct an auxiliary hypergraph to facilitate efficient partitioning and introduce a framework for local motif-based clustering.  Extensive experiments on real-world datasets demonstrate the effectiveness of our framework and provide a comparative analysis of the two proposed clustering strategies in terms of clustering quality and computational efficiency.

\end{abstract}

\section{Introduction}



Graphs are widely employed in many areas, including data mining and network analysis, for modeling complex pairwise relationships. While traditional graphs effectively represent pairwise interactions, they fall short when it comes to modeling higher-order relationships often found in real-world data.
To address this, \emph{hypergraphs} extend the capabilities of graphs where each hyperedge can connect multiple nodes simultaneously, thereby capturing complex multi-way interactions that conventional graphs cannot and making them well-suited for domains where higher-order information is essential. Numerous important real-world systems, including social networks and distributed databases, can be effectively represented using hypergraph models.

Network-based models have become indispensable tools for analyzing complex systems across diverse domains. These models, particularly in the form of graphs, offer a rich suite of analytical techniques that help uncover complex patterns of connectivity, both at local and global scales. One of the central tasks in understanding the structure and behavior of such networks is \emph{community detection} and \emph{local clustering}~\cite{GiNe}. The study of local clustering, which aims to identify a tightly connected community around a particular seed node, has been extensively researched in the context of standard graphs. However, extending this task to hypergraphs appears to be non-trivial and it is still quite an underexplored area. Existing methods either rely on transforming hypergraphs into pairwise graph structures or directly apply graph-based algorithms which often fail to fully leverage the high-order structural information inherent in hypergraphs, resulting in inaccurate assessments of cluster quality.

This research introduces a novel approach to local clustering in hypergraphs by developing a new measure tailored to their unique structural properties. Our model leverages higher-order motifs, which are small, connected subgraphs where vertices are linked through interactions among multiple nodes, to more accurately evaluate cluster quality. Building upon this measure, we present a local clustering algorithm designed specifically for hypergraphs. Experimental results on real-world datasets demonstrate that our method achieves superior performance both in accuracy and computational efficiency than standard method.

Hypergraph-based local clustering has broad applications across various fields, including bioinformatics~\cite{JiWaWaXi}, disease gene and drug discovery~\cite{SuWaZhCoetal}
, aiding in tasks such as identifying disease genes or protein complexes. 
These applications underscore the importance of developing hypergraph-specific tools that fully exploit their modeling capabilities. Our work addresses this need by providing a theoretically grounded and practically effective framework for local clustering in hypergraphs. Our study focuses on the key foundational concepts that have most directly influenced our approach. In particular, it builds upon developments in graph and hypergraph clustering, higher-order motif analysis, and both local and partitioning-based clustering techniques.

\subsection{Related Work}

Local clustering is a well-studied area both for graphs and hypergraphs, where recent advances focus on motif-based methods that capture higher-order structures beyond traditional edge-based clustering.

\subsubsection{Local Clustering on Graphs} 
Clustering methods include partitioning, hierarchical, density-, grid-, model-, and motif-based approaches~\cite{Faetal}. Local clustering has been extensively studied both theoretically~\cite{AnChLa} and empirically~\cite{LeLaDaMa}, resulting in statistical~\cite{KlGl}, numerical~\cite{LiHeBiHo}, and combinatorial techniques~\cite{ FoLiGlMa}. Recent advances focus on motifs, frequently recurring subgraphs that capture key network structures~\cite{ChFaSc, YiBeLeGl}. Motifs are widely applied in fields like biology~\cite{Alon}, medicine~\cite{ChQuCaZhLiLiLiHeFeJi}, social networks~\cite{HoHaCuPi}, and finance~\cite{SaDiGaSq}, helping characterize network entities and highlighting the value of motif-based clustering for understanding graph structure~\cite{Laetal}.

A central challenge in these approaches is the computational cost of motif detection~\cite{LiHuWaLaHu}. Benson \textit{et al.}~\cite{BeGlLe} introduced an efficient motif-based clustering method, inspiring further optimizations~\cite{MaCaHeChCaLi}. Tsourakakis \textit{et al.}~\cite{TsPaMi} studied triangle motifs and proposed triangle motif conductance for community evaluation. 
Recently, Xia \textit{et al.}~\cite{XiYuLiLiLe} proposed \emph{CHIEF} (Clustering with HIgher-ordEr motiFs), which reduces computational overhead by partitioning networks with $k$-connected graphs before motif clustering.


These studies show that local motif clustering, especially using motif conductance, effectively identifies high-quality communities and reconstructs ground-truth structures. However, as this approach is recent, most methods adapt statistical or numerical edge-based techniques or use simple combinatorial algorithms. Chhabra \textit{et al.}~\cite{ChFaSc} proposed advanced combinatorial methods with two algorithms, one graph-based and one hypergraph-based, that perform local motif clustering via partitioning.


\subsubsection{Local Clustering on Hypergraphs}
Hypergraph clustering has recently gained traction in computer vision and machine learning~\cite{AgLiZePeKrBe} due to its ability to model complex relationships beyond pairwise links~\cite{MaMaBhEpRo}. Since higher-order structures reveal patterns missed by traditional edge-based methods, several approaches~\cite{TsPaMi} have been developed for higher-order and hypergraph clustering. Higher-order motifs fit naturally in hypergraphs, with hyperedges representing motif instances.


In~\cite{BeGlLe} and~\cite{TsPaMi}, hypergraph data is transformed into weighted graphs or motif adjacency matrices, where entries count motifs shared by node pairs, enabling spectral clustering. To enhance scalability in tasks like face and motion clustering, Pukrait \textit{et al.}~\cite{PuChSaSu} introduced a guided sampling method for large hyperedges. 


To address the challenge of local clustering within hypergraphs, a sweep cut method using hypergraph Personalized PageRank (PPR) was introduced~\cite{TaMiIkYo}. The HyperGo algorithm~\cite{LiHeZh18} adapts the NIBBLE local partitioning method~\cite{SpTe} for hypergraphs. The Capacity Releasing Diffusion (CRD) algorithm~\cite{IbGl} offers a non-spectral approach preserving local structures, with its extension HG-CRD operating directly on hyperedges and improving clustering via motif conductance. Recently,~\cite{WeYaLuZhQiZh} proposed a degree-based conductance metric and a greedy algorithm for efficient local clustering in hypergraphs.


\medskip
Although motif-based clustering has been extensively studied in graphs, its direct application to hypergraphs remains limited. Existing hypergraph clustering methods either rely on graph transformations or overlook higher-order structures. Our work bridges this gap by explicitly integrating higher-order motifs within hypergraph clustering algorithms.

\subsection{Our Contribution} 

As previously noted, current methods either adapt graph-based local clustering techniques directly to hypergraphs or reduce hypergraphs to standard graphs, thereby neglecting their intrinsic higher-order structure and often resulting in low-quality clustering. To address this limitation, we introduce a novel approach specifically designed to identify high-quality local clusters within hypergraphs.

In previous studies, higher-order motifs have predominantly been defined as triangles, $k$-connected subgraphs, or small induced patterns known as graphlets. In contrast, our work employs a distinct class of higher-order motifs specifically defined for hypergraphs, as originally introduced in~\cite{LoMuMoBa} (for more details please refer to Section~\ref{HOM}). Analogous to traditional pairwise network motifs, higher-order motifs for hypergraphs hold significant potential for enabling applications across various domains. Overall, our work emphasizes the strength of higher-order motifs in hypergraphs and introduces a framework for local hypergraph clustering via extracting higher-order fingerprints within hypergraphs. 

Building upon this foundation, we extend the framework  developed by  Chhabra \textit{et al}{.}~\cite{ChFaSc} to formulate a local clustering strategy for hypergraphs based on the hypergraph-specific motifs while simultaneously optimizing motif conductance. 
Our main contributions are summarized as follows:

\begin{itemize}
\item Building on the concept of higher-order motifs of hypergraphs, we design a novel local clustering algorithm for hypergraphs, introducing a practical framework that avoids reducing hypergraphs to graphs, thereby maintaining their native multi-way interaction structure.
\item We propose two alternative strategies for identifying local clusters around a given seed hyperedge:
\begin{enumerate}
\item \textbf{Core-based Method:} We identify a localized set of nodes (ball $B$) around the seed using hypergraph core decomposition, followed by motif enumeration and construction of an auxiliary hypergraph $H_{\mu}$. We then perform partitioning to extract the target local cluster.
\item \textbf{BFS-based Method:} As an alternative, we use a breadth-first search to define $B$, while the remaining steps of the process remain unchanged.
\end{enumerate}
\item We conduct an extensive experimental evaluation on real-world hypergraph datasets, providing a significant comparison of core-based versus BFS-based local clustering strategies for hypergraphs based on higher-order motifs.
\end{itemize}

Our experiments showcase the practical value of the proposed framework, highlighting its ability to identify high-quality local clusters that preserve higher-order connectivity patterns in real-world hypergraph datasets.




\section{Preliminaries}
Let $H= (V, E)$ be an undirected hypergraph, where $V$ is the set of nodes, $E$ is the set of hyperedges (or nets), and each hyperedge (net) $e\in E$ is a subset of $V$.
Without loss of generality, we assume $V= \{0, . . . , n-1\}$. We also assume that there are no multiple or self hyperedges. Let $n = |V|$ and $m= |E|$. Let $c : V \to R$ be a node-weight function, and let $\omega: E \to R$ be an edge-weight function. We generalize the functions $c$ and $\omega$ to sets such that $c(V')=\sum_{v\in V'} c(v)$ and $\omega(E')=\sum_{e\in E'} \omega(e)$. For any hypergraph $H=(V,E)$, two nodes $v$ and $u$ are said to be adjacent if there exists a hyperedge $e \in E$ that contains both $v$ and $u$. Otherwise, we say that $v$ and $u$ are not adjacent. For any node $v$ in a hypergraph $H=(V,E)$, the set $N[v] = \{u \in V : u \text{ is adjacent to } v\}\cup \{v\}$ is called the closed neighborhood of $v$ in $H$ and each node in the set $N[v] \setminus \{v\}$ is called a neighbor of $v$. 
The open neighborhood of the node $v$ is the set $N[v] \setminus \{v\}$. We generalize the notations 
of neighborhood to sets, such that if $S  \subseteq V$ then $N(S) = \bigcup_{v\in S}N(v)$ and $N[S] = N(S) \cup S$.
Let $H =(V,E )$ be a hypergraph, and let $V'\subseteq V$.
The  subhypergraph $H[V']$ induced by $V'$   is the hypergraph $H[V'] = (V',E')$ where $E'= \{e \in E : e \subseteq V'\}$. 
 $E'$ can be represented as a multi-set. Moreover, according to the remark above, we can add, if needed, the empty set. 

A 1-cycle (or loop) is a hypergraph consisting of one edge which contains a repeated node. A 2-cycle is a hypergraph consisting of two edges which intersect in at least 2 nodes. For $j\geq 3$, a $j$-cycle is a hypergraph with $j$ edges which can be labelled $e_1,e_2, \cdots ,e_j $ such that there exist distinct nodes $v_1,\cdots,v_j$ where $ v_i \in e_i \cap e_{i+1}$ for $i = 1,\cdots,j $(where $e_j+1 \equiv e_1$). A (Berge) path in $H$ consists of a sequence $ v_0,e_1,v_1,e_2,\cdots,e_j,v_j$ where $v_0,v_1,\cdots,v_j$ are distinct nodes, $e_1,\cdots,e_j$ are distinct edges, and $v_{i-1},v_i \in e_i$ for all $i = 1,\cdots,j$. A hypergraph $H = (V, E)$ is connected if for every pair of nodes $u, v \in V$, there exists a (Berge) path connecting $u$ and $v$. 

Let $\overline{V'} = V \setminus V'$ be the complement of a set $V' \subseteq V$ of nodes. 
A higher-order motif can be considered as a connected hypergraph, where the order of a motif refers to the number of nodes involved.
Let $\mu$ be a higher-order motif.
Enumerating the higher-order motifs $\mu$ in a hypergraph $H$ consists of building the collection $M$ of all occurrences of $\mu$ as a subhypergraph of $H$. The combinatorial explosion of higher-order motifs makes their storing and indexing in memory (required steps for counting their occurrences in empirical hypergraphs and evaluating their over- or under-expression) intractable for high orders. 

The degree of a node in a hypergraph is defined as the number of hyperedges including the node. Similarly, the degree of a hyperedge in a hypergraph is defined as the number of nodes included by the hyperedge. Let $ d(v)$ be the degree of node $v$ and $\Delta$ be the maximum degree of $H$. 
Let $d_{\mu}(v)$ be the number of higher-order motifs $\mu \in M$ which contain $v$. We generalize the notations $d(\cdot)$ and $d_{\mu}(\cdot)$ to sets, such that $d(V')=\sum_{v \in V'}d(v)$ and $d_{\mu}(V')=\sum_{v \in V'}d_{\mu}(v)$.



A $k$-way partition of a (hyper)graph $\mathcal{H}$ is a partition of its node set into $k$ blocks $P = \{V_1, \ldots, V_k\}$ such that:
$$\bigcup_{i=1}^k V_i = V, V_i \neq \emptyset \text{ for } 1 \le i \le k, \text{ and } V_i \cap V_j = \emptyset \text{ for } i \neq j.$$


The cut-net which consists of the total weight of the nets crossing blocks,
$$\sum_{e \in E'} w(E'),$$
in which 
$$E' := \{e \in E : \exists i, j \text{ | } e \cap V_i \neq \emptyset, e \cap V_j \neq \emptyset, i \neq j\}.$$

In the local hypergraph clustering problem, a hypergraph $H = (V , E)$ and a seed subset of nodes $ S \subseteq V$ are taken as input and the goal is to detect a well-characterized cluster (or community) $C \subset V$ containing $S$. A high-quality cluster $C$ usually contains nodes that are densely connected to one another and sparsely connected to $\overline{C}$. There are many functions to quantify the quality of a cluster, such as modularity~\cite{Ulrik} and conductance~\cite{Ravi}. Local higher-order motif hypergraph clustering is a generalization of local hypergraph clustering where a higher-order motif $\mu$ is taken as an additional input and the computed cluster optimizes a clustering metric based on $\mu$. To a given higher-order motif $\mu$, recognize a cluster of nodes, $C$, with the following two constraint conditions. (1) Nodes in $C$ should participate in many higher-order motifs with the same structure to $\mu$. (2) The set $ C$ should avoid cutting $\mu$. This occurs when only a subset of the nodes from a higher-order motif are in the set $C$. We find a cluster that minimizes the ratio $\phi_{\mu}(C)$, where in $\phi_{\mu}(C)$ is defined as $\phi_{\mu}(C) = cut_{\mu}(C, \overline{C})/min[d_{\mu}(C),d_{\mu}(\overline{C})]$ where $\overline{C}$ refers to the complement of $C$,  $cut_{\mu}(C, \overline{C})$ represents the number of higher-order motif $\mu$ with at least one node in $C$ and one in $\overline{C}$, and $d_{\mu}(C)$ is the number of higher-order motifs $\mu$ that reside in $C$. $\phi_{\mu}(C)$ refers to the motif conductance of $C$ with respect to $\mu$~\cite{Bretto, XiYuLiLiLe, ChFaSc}.

\section{Local Hypergraph Clustering: Higher-Order Motif-Based}


In this section, we present our local clustering algorithm for hypergraphs, which consists of four main phases.
Given a hypergraph $H=(V,E)$ and a seed hyperedge $e$, the algorithm proceeds as follows. First, a set of nodes called ball $B$ around the seed hyperedge $e$ (i.e., containing the hyperedges $e$), is selected through either using the core-decomposision method or the BFS method. Secondly, the collection $M$ of higher-order motif occurrences is enumerated, where each occurrence contains at least one node in $B$, and we specifically use higher-order motifs of order $3$. Then, a new hypergraph $H_{\mu}$ is constructed to facilitate motif conductance computation. Finally, $H_{\mu}$ is partitioned into two blocks using the  KaHyPar
hypergraph partitioning algorithm from~\cite{ScHeHeMeSaSc}, and the result is mapped back to $H$ as a local clustering around the seed $e$.

Since hypergraph partitioning minimizes the cut-net value rather than optimizing motif conductance directly, the partitioning phase is repeated $\beta$ times with different imbalance constraints, selecting the best result. These parameters, whose values are specified in Section \ref{sec:experimental_setup}, are selected based on experimental evaluation. Additionally, when using BFS method to compute the ball $B$ the first phase is repeated $\alpha$ times to enhance exploration around $e$.

\subsection{Phase One: Locating the Ball}

As mentioned earlier, the first step in our algorithm is to identify a ball $B$ around the seed $e$, using either the core decomposition method or the breadth-first search (BFS) approach. Below, we describe each method in detail.

\subsubsection{Core around the Seed Hyperedge}
In graph theory, the concept of $k$-core decomposition is a tool for analyzing the structure and connectivity of networks. A $k$-core of a graph is a maximal subgraph in which every node has at least $k$ neighbors within the subgraph. This concept has been widely used in various applications, such as community detection, network visualization, and understanding the robustness of networks. However, when moving from graphs to hypergraphs, the definition and computation of cores become more complex due to the presence of hyperedges that can connect multiple nodes simultaneously.

Various definitions and algorithmic approaches have been developed to compute cores in hypergraphs, adapting traditional graph-based methods to the unique structure of hypergraphs while addressing their complexities~\cite{KhBaSrTh, LiZhHuXu,LuZhZhSt,MoPeMi}. One such approach is the neighborhood-based core decomposition, which extends the concept of $k$-core decomposition to hypergraphs by considering the number of neighbors each node has within a subhypergraph. This approach ensures that nodes maintain a minimum level of connectivity, even in the presence of hyperedges that may span multiple nodes~\cite{ArKhRaGh}. In the context of hypergraphs, the nbr-$k$-core $H[V_k]=(V_k, E[V_k])$ of a hypergraph $H=(V,E)$, denoted as $H_k$, is defined as the maximal (strongly) induced subhypergraph where every node $u\in V$ has at least $k$ neighbors in $H[V_k]$. The maximum core is the largest value of $k$ for which $H_k$ is non-empty. The core-number $c(v)$ of a node $v \in V$ is the largest $k$ such that $v\in V_k$  and $v\not\in V_{k+1}$. Thus, the core decomposition of a hypergraph assigns to each node its core number, capturing the nested hierarchy of connectivity in the hypergraph.

The definition of the nbr-$k$-core ensures that nodes maintain a minimum number of neighbors within the subhypergraph. However, computing cores in hypergraphs is more challenging than in ordinary graphs due to the presence of hyperedges that can connect multiple nodes. Specifically, in hypergraphs, the neighborhood of a node is determined not just by pairwise connections (as in graphs) but by its participation in hyperedges, which may involve multiple nodes. Also, the strongly induced subhypergraph condition ensures that only hyperedges fully contained within $V_k$ are considered, preserving the neighborhood structure and maintaining the coreness condition for all nodes in the subhypergraph.

\subsubsection{Ball around Seed Hyperedge}
First, we extend the BFS algorithm to hypergraphs. Unlike standard graphs, where an edge connects only two nodes, a hyperedge can simultaneously link multiple nodes. In hypergraphs, traversal occurs through hyperedges, potentially introducing several new nodes at once. Specifically, the BFS algorithm for hypergraphs begins with a given seed (either a node or a hyperedge i.e., multiple nodes) and explores all reachable nodes level by level. Instead of processing pairwise edges as in conventional graphs, it navigates through hyperedges, each capable of connecting multiple nodes. By iterating over these hyperedges and enqueueing newly discovered nodes, the algorithm ensures that all nodes are visited in the smallest possible number of steps. The complexity of BFS in hypergraphs is influenced by the number of hyperedges $E_H$ and their respective sizes $\vert e\vert =k$ resulting in a time complexity of $O\left(V + \sum_{e \in E_H} |e|\right)$, where $\vert e\vert $ denotes  the number of nodes in a given hyperedge $e$. 

The method proposed in~\cite{ChFaSc} selects $B$ using a breadth-first search (BFS) starting from a seed. Specifically, the BFS tree is expanded up to $l$ layers, and all nodes within these layers are included in $B$. To enhance exploration, the algorithm is repeated $\alpha$ times with varying values of $l$. Two special cases are addressed: if $B$ is too small, at least one repetition ensures that it contains at least $100$ nodes by expanding layers when necessary, as suggested in~\cite{LeLaDaMa}, unless the dataset itself has fewer than 100 nodes. Additionally, if the entire BFS tree consists of at most $l$ layers, the algorithm returns a single set $B$. This approach increases the probability of capturing a well-defined community containing the seed. The BFS selection is performed within the subhypergraph induced in the connected component that contains the seed, specifically within its closed neighborhood $N[B]$. Once $B$ is selected, subsequent computational steps operate only on $B$, its hyperedges, and motif occurrences, significantly reducing the computational complexity. The parameters $\alpha$ (number of repetitions) and $l$ (layers per repetition) can be adjusted as needed.\\

\subsection{Phase Two: Motif Enumeration}
\label{HOM} 
Motif analysis has become a crucial tool in network science for identifying characteristic patterns at the microscale. Higher-order network motifs are defined as small, connected structures of higher-order interactions that occur in a given hypergraph with significantly greater frequency than in a suitably randomized counterpart. Similar to traditional motif analysis, the key steps in higher-order motif analysis include: (i) counting the occurrences of each higher-order motif within the network, (ii) comparing these frequencies to those found in a null model, and (iii) assessing their over- or under-representation using statistical measures. A comprehensive discussion of algorithms and tools used to extract and analyze higher-order motifs is provided in~\cite{LoMuMoBa}. 

\begin{figure}[t!]
\centering
\includegraphics[width=0.85\linewidth]{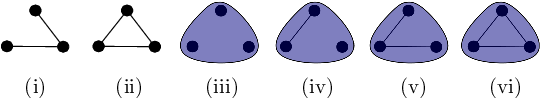}
\caption{All six patterns of higher-order interactions among three nodes. Blue-shaded triangles indicate higher-order interactions, while black lines signify pairwise interactions.}
\label{fig:HOMotif3}
\end{figure}

The analysis of higher-order motifs of order 3 across different domains reveals the relative structural significance of certain interaction patterns (Figure~\ref{fig:HOMotif3}). As studied in~\cite{LoMuMoBa} the pairwise triangle (motif (ii)) is consistently highly overrepresented in all domains. However, the most notable variations across domains arise from motifs that involve a 3-hyperedge combined with at least one dyadic edge. In social and technological networks, motif (vi) comprising a 3-hyperedge and a triangle of dyadic edges is strongly overrepresented, suggesting that entities engaging in group interactions also tend to form individual connections. In co-authorship networks, motifs (iv) and (v) are the most overrepresented; these structures consist of a 3-hyperedge along with one or two dyadic edges. This pattern indicates a possible hierarchical organization, where not all nodes interact equally in pairs. For example, a research leader may co-author papers with students and postdocs, while the latter do not collaborate independently without the leader’s involvement.

We now present the motif-enumeration phase of our approach and discuss its implementation. We efficiently solve this step for higher-order motifs of order 3, based on the size of the subhypergraph induced in $H$ by the closed neighborhood $N[B]$ of $B$. Higher-order motifs of order 3 have numerous important applications in network analysis and clustering as mentioned earlier. Without loss of generality, our focus is specifically on higher-order motifs of order 3. However, the overall approach can be adapted to accommodate more complex higher-order motifs.  

An algorithm for enumerating higher-order motifs of order 3 and 4 was introduced by~\cite{LoMuMoBa}. Broadly speaking, this algorithm follows a hierarchical approach. It begins by iterating over all hyperedges of size $k$ that can directly induce a motif meaning that a hyperedge of size $k$ provides all the necessary nodes to form a motif of order $k$. The process then continues by examining hyperedges of smaller sizes until it reaches traditional dyadic edges. Since hyperedges smaller than $k$ cannot directly induce a motif, the algorithm follows a similar strategy to~\cite{Wer} by selecting the remaining nodes based on the neighborhood of the subhypergraph. Once $k$ nodes have been selected, the algorithm efficiently constructs their induced subhypergraph by iterating over the power set of these $k$ nodes, which results in $2^k$ possible hyperedges. It then retains only those hyperedges that exist in the original hypergraph.

In the higher-order motif-enumeration phase of our algorithm, we utilize the algorithm proposed in~\cite{LoMuMoBa}, but only on the subhypergraph induced in $H$ by $N[B]$. This approach is sufficient to identify all higher-order motifs of order 3 that include at least one node from $B$, as demonstrated by transformation (1) in Figure~\ref{fig:hyperedges}.

\begin{figure}[t!]
\centering
\includegraphics[width=0.99\linewidth]{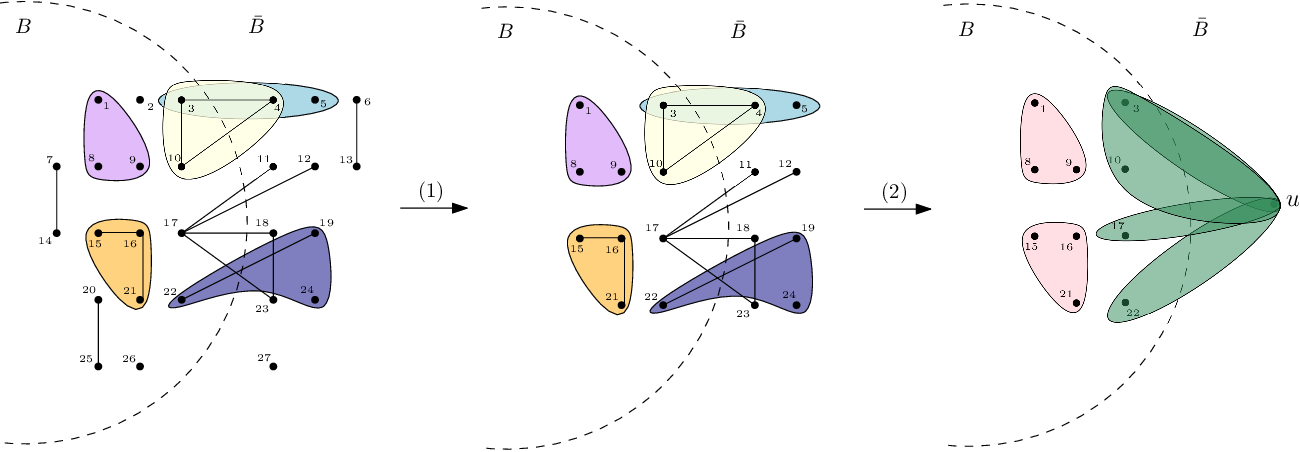}
\caption{An example of the higher-order motif-enumeration and auxiliary hypergraph construction phases of our approach. On the left, the nodes of $H$ are divided into sets $B$ and $\overline{B}$. In the center, higher-order motif occurrences that include at least one node in $B$ are enumerated. On the right,  $H_{\mu}$ is constructed by transforming higher-order motifs into hyperedges and contracting $\overline{B}$ into a single node $u$.
}
\label{fig:hyperedges}
\end{figure}

\subsection{Phase Three: Construction of Auxiliary Hypergraph $H_{\mu}$}

In this section, we provide a conceptual description of an auxiliary hypergraph $H_{\mu}$ that we construct to support the next phase of our algorithm, which is partitioning. We explain how to build this hypergraph using the method proposed in~\cite{ChFaSc}. We demonstrate that $H_{\mu}$ can be constructed in linear time relative to the number of nodes in $B$ and the higher-order motifs in $M$ (the set of higher-order motif occurrences identified in the previous step). Additionally, we discuss the advantages and limitations of this hypergraph-based approach.

Our auxiliary hypergraph $H_{\mu}$ is constructed through two main steps. First, we define a hypergraph with $V$ as the set of nodes and a set $E$ of hyperedges. For each higher-order motif in $M$, $E$ includes a hyperedge whose nodes correspond to the nodes of that motif. Next, we merge all nodes in $\overline{B}$ into a single node $u$, and replace parallel hyperedges with a single hyperedge whose weight is the sum of the weights of the removed parallel hyperedges. Formally, this hypergraph is defined as $H_{\mu} = (B \cup \{u\}, E)$, where the set $E$ consists of a hyperedge $e$ associated with each higher-order motif occurrence $H' = (V', E') \in M$. If $V'\subseteq B$, then $e = V'$; otherwise, $e = V'\cap B \cup \{u\}$. In the first case, the hyperedge has weight 1, while in the second case, the hyperedge’s weight equals the number of higher-order motif occurrences in $M$ it represents. For the partitioning process, we assign equal weights to the nodes, as unbalanced node weights could affect the performance of hypergraph partitioning algorithms, commonly referred to as partitioners. However, the correct weights are used when evaluating the objective. In practice, the construction of $H_{\mu}$ involves instantiating the nodes in $B\cup \{u\}$ and the hyperedges in $E$. Assuming the number of nodes in $\mu$ is constant, the hypergraph is constructed in $O(|S| + |M|)$ time and requires $O(|S| + |M|)$ memory. The construction of $H_{\mu}$ is depicted in transformation (3) of Figure~\ref{fig:steps} and demonstrated with an example in transformation (2) of Figure~\ref{fig:hyperedges}.

By examining the relationship between $H$ and $H_{\mu}$, we can categorize the components into three distinct groups. The first group consists of the nodes in $B$ and the higher-order motifs where all nodes are in $B$, which are represented in $H_{\mu}$ without any contraction, as individual nodes and hyperedges. The second group includes the nodes in $\overline{B}$ and the higher-order motifs where all nodes are in $\overline{B}$, which are compactly represented in $H_{\mu}$ as the contracted node $v$. The third group consists of higher-order motifs that contain nodes from both $B$ and $\overline{B}$, and these higher-order motifs are represented in $H_{\mu}$ as hyperedges that include individual nodes from $B$ along with the node $u$. In summary, this auxiliary hypergraph offers a compact representation of the entire hypergraph $H$, emphasizing relevant information for local higher-order motif clustering in two ways: hyperedges are omitted while higher-order motifs are made explicit, and global information is abstracted while local information is preserved in detail. Theorem~\ref{thm1} demonstrates that the cut-net of a partition of $H_{\mu}$ directly corresponds to the motif-cut of an equivalent partition in $H$, provided our motif enumeration step is accurate. Additionally, Theorem~\ref{thm2} shows that the motif conductance of this equivalent partition of $H$ can be directly computed from $H_{\mu}$, assuming $d_{\mu}(B) \leq d_{\mu}(\overline{B})$. The assumption $d_{\mu}(B) \leq d_{\mu}(\overline{B})$ is reasonable because $B$ is typically smaller than $\overline{B}$. Although enumerating the higher-order motifs in $\overline{B}$ may not be ideal for a local clustering algorithm, we have verified that this assumption holds throughout all our experiments.

\begin{Theorem} \label{thm1}

Any $k$-way partition $P$ of our auxiliary hypergraph $H_\mu$ corresponds to a unique $k$-way partition $P'$ of $H$ where the cut-net of $P$ is identical to the motif-cut of $P'$, provided that the motif enumeration step is exact.
\end{Theorem}

\begin{proof}
Consider the auxiliary hypergraph $H_\mu$ but now without the replacement of the parallel hyperedges with a single hyperedge whose weight is the sum of the weights of the removed parallel hyperedges.
Observe that by the construction of $H_\mu$, we get a correspondence between the nodes of $H_\mu$ and the nodes of $H$.
Thus, any partition $P$  of $H_\mu$ aligns with a partition $P'$ of $H$, with the corresponding nodes form the same blocks.
Next, there is no motif that is totally contained in $S'$ can be cut in $P'$, since $B'$ is just a node in $H_\mu$. 
On the other hand, every motif between $B$ and $B'$ correspond to a hyperedge in $H_\mu$ with the sense that both of them contain the same nodes in $H$ and $H_\mu$, respectively. 
Therefore, a motif occurrence of $H$ is cut in $P'$ if and only if the corresponding hyperedge in $H_\mu$ is cut in $P$.
Finally, it can be extended even if we consider the weight on the hyperedges instead of the parallel hyperedges.
\end{proof}

\begin{figure}[t!]
\centering
\includegraphics[width=0.99\linewidth]{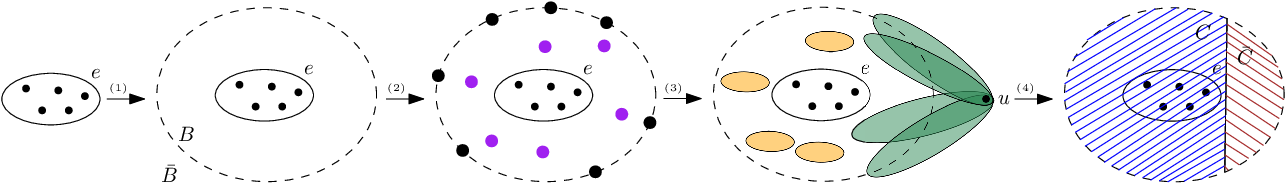}
\caption{Illustration of the steps in our proposed algorithm. (1) Given a seed hyperedge $e$ and a hypergraph $H$, a ball $B$ surrounding $e$ is selected. (2) Higher-order motif occurrences of $\mu$ that include at least one node in $B$ are enumerated. (3) An auxiliary hypergraph $H_{\mu}$ is then constructed by transforming higher-order motifs into hyperedges and contracting $\overline{B}$ into a single node $u$. (4) The auxiliary hypergraph is partitioned into two blocks using a multi-level hypergraph partitioner, and the partition of $H_{\mu}$ that does not contain $u$ is mapped back as a local cluster around the seed in $H$.}
\label{fig:steps}
\end{figure}

\begin{Theorem} \label{thm2}
Given a 2-way partition $P = (C,\overline{C})$ of our auxiliary hypergraph $H_\mu$ with $ u \in \overline{C}$, the motif conductance $\phi_{\mu}(C')$ of the corresponding 2-way partition $P' = (C',\overline{C'})$ of $H$ is defined as the ratio of the cut-net of $P$ to $\lceil \mathcal{W}(C)$, assuming an exact motif enumeration step and $d_\mu(S) \leq d_\mu(\overline{S})$.
\end{Theorem}

\begin{proof}
From Theorem~\ref{thm1}, the motif-cut of $P'$ can be substituted by the cut-net of $P$ in the numerator of the definition of $\phi_{\mu}(C')$. 
Thus, we need only to show that the denominator of $\phi_{mu}(C')$, that is $\min[d_{\mu}(C),d_{\mu}(\overline{C})]$, is equal to  $\lceil \mathcal{W}(C)$.
By the construction of $H_\mu$, the values of $d_{\mu}(C')$ and $\lceil \mathcal{W}(C)$ are the same.
Moreover,  $ u \in \overline{C}$ says that $S \subseteq C'$ and $C'\subseteq S$, which means that $d_{\mu}(\overline{S}) \leq d_{\mu}(\overline{C'})$ and $d_{\mu}(C') \leq d_{\mu}(S)$.
Thus, $\lceil \mathcal{W}(C) = d_{\mu}(\overline{C'}) \leq  d_{\mu}(S) \leq d_{\mu}(\overline{S}) \leq d_{\mu}(\overline{C'})$, due to $d_{\mu}(S) \leq d_{\mu}(\overline{S})$.
\end{proof} 

\subsection{Phase Four: Partitioning}

This section outlines the hypergraph partitioning phase of our local higher-order motif clustering algorithm. We describe the hypergraph partitioning algorithms employed and elaborate on the strategies used to ensure the feasibility of the obtained solution while enhancing its quality. Additionally, we provide insights into the computational complexity of the partitioning phase. Our algorithm's partitioning phase involves a two-way partitioning of $H_\mu$, which is performed using the multi-level hypergraph partitioner KaHyPar~\cite{ScHeHeMeSaSc}.

This partitioner employs advanced algorithms to efficiently generate low-cut partitions of hypergraphs. As previously established, any two-way partition of $H_\mu$ inherently corresponds to a community in $H$. However, our objective is to derive a consistent partition of $H_\mu$, which we define as a partition where the seed and the contracted node $u$ reside in separate blocks. This consistency ensures that the resulting partition corresponds to a local community in $H$ that includes the seed and remains entirely contained within the ball $B$. Maintaining this criterion is crucial, as the nodes in $\overline{B}$ have not been explored by our algorithm and are located farther from the seed than those in $B$. Notably, the two-way partition illustrated in Figure~\ref{fig:steps} produces a consistent local community according to our definition.  

Additionally, KaHyPar supports hypergraph partitioning with fixed nodes. To ensure block feasibility within our algorithm, we enforce that the seed is assigned to the block that excludes $u$ after partitioning is completed (prior to computing motif conductance). Despite its simplicity, this approach has shown minimal impact on solution quality in preliminary evaluations compared to a fixed-node-based method. Apart from ensuring a consistent local clustering in the original hypergraph, our solution aims to achieve minimal motif conductance. However, since KaHyPar is a randomized algorithm that primarily minimizes the cut value while enforcing a strict balancing constraint, it does not directly optimize motif conductance. To enhance our results, we explore various combinations of cut-net and imbalance by executing the partitioning process $\beta$ times with randomized balancing constraints for each constructed auxiliary hypergraph, where $\beta$ is a tuning parameter. For each partition obtained, our algorithm evaluates the motif conductance of the corresponding local cluster, as outlined in Theorem~\ref{thm2}, and retains the partition with the lowest motif conductance.

KaHyPar exhibits near-linear runtime in practical scenarios. Although it employs advanced algorithms and data structures to minimize the cut value introducing higher constant factors in its runtime complexity, the motif enumeration phase can become the dominant computational bottleneck in our local clustering algorithm, particularly for higher-order motifs of order greater than three. Additionally, KaHyPar may experience performance slowdowns as the number of hyperedges in the auxiliary hypergraph grows significantly. In such cases, runtime efficiency can be further enhanced by leveraging parallelized tools like Mt-KaHyPar~\cite{GoHeSaSch}. However, parallelization falls outside the scope of this study.

\begin{table}[t!]
    \centering
    \begin{tabular}{l|c|c}
        \hline
        \textbf{Dataset} & \textbf{$|V|$} & \textbf{$|E|$} \\
        \hline
        DBLP & 1924991 & 2466799 \\
        coauth-MAG-History & 1014734 & 895439 \\
        Enron & 4423 & 5734 \\
        Email\_Eu & 998 & 25027 \\
        Contact-primary-school & 242 & 12704 \\ 
        Conference & 403 & 10541 \\ 
        Contact-high-School & 327 & 7818 \\
        Hospital & 75 & 1825 \\     
        \hline
    \end{tabular}
    \caption{Size of hypergraphs used in the experiments.}
    \label{tab:dataset_stats}
\end{table}

\section{Experimental Evaluation}
We implemented our algorithm in Python using the KaHyPar framework which was implemented originally in C++, leveraging the public libraries provided by KaHyPar~\cite{ScHeHeMeSaSc} for hypergraph partitioning.

\subsection{Experimental Setup} \label{sec:experimental_setup}
Through the evaluation of our proposed model local partitioning on hypergraphs, we employed 9 diverse datasets from different domains and application areas to examine the efficiency of our approach in terms of computational performance, scalability, and accuracy across real-world scenarios. For our evaluation of real-world higher-order systems, we collected a number of freely available networked datasets. The datasets come from a variety of domains: sociology (proximity contacts), technology (e-mails), biology (gene/disease, drugs) and co-authorship. The description of
each dataset is reported in Supplementary Note 1. Additionally, it is of major importance to mention that the data-set is the one provided by https://www.cs.cornell.edu/~arb/data/. We implemented our proposed model in Python and performed the experiments on a commodity computer MacBook Pro Apple M3 with 16Gb of RAM.

In Table~\ref{tab:performance}, we compare the two alternative approaches for identifying local clusters core-based and BFS-based methods. More specifically, we use the following parameters for our algorithm: $\alpha = 3$, $\beta = 80$. In the current implementation, the appropriate ball size of $S$, was selected by identifying the first breadth-first search (BFS) layer that resulted in more than 100 nodes, as previously described. Subsequently, two additional consecutive BFS layers were computed. Among the three resulting BFS balls constructed through this process, the one achieving the optimal motif conductance score was selected as the final output.

\subsection{Analysis of the Methods}


The core-based method is based on $k$-core decomposition, ensuring that selected nodes maintain a minimum level of connectivity, which leads to structurally cohesive clusters. 
While, the BFS-based method follows a traversal strategy that expands clusters based on neighborhood exploration, which may include weaker connections.
Table~\ref{tab:performance} presents a comparative analysis of the two methods across a range of real-world datasets. A key metric examined is motif conductance ($\phi_\mu$), which evaluates how well motifs are preserved within the resulting clusters. Across most datasets, the core-based method consistently achieves lower motif conductance values (e.g., 0.011 with Core ball vs. 0.029 with BFS ball in DBLP graph), showing that this method tends to preserve high-order motif participation within the cluster.

Another interesting observation is the stability of conductance values in the core-based method. In dataset of Contact-high-School graph, both methods produce clusters of comparable size and conductance, but the core-based approach achieves this with less computational time (21s vs. 66s), showing that core decomposition can sometimes provide not only better but also more efficient results.
Regarding cluster size, the results reveal no consistent dominance of one method over the other. 
In some datasets like DBLP and coauth-MAG-History, the core-based method produces larger clusters, likely because the area around the seed is densely connected and supports wider inclusion. On the other hand, in datasets such as Email-Eu and Enron, the BFS-based method results in larger clusters, probably because it includes more loosely connected nodes that would not be part of a core.

Overall, these results support the conclusion that while both methods can be effective depending on the graph structure, the core-based method offers a stronger guarantee of higher-order motif structure, especially in dense or highly structured hypergraphs. Its ability to preserve high-order connectivity and guide more meaningful motif-based partitioning leads to more structurally connected and internally consistent clusters.

\begin{table}[t!]
\centering
\label{tab:performance}
\begin{tabular}{lcccccc}
\hline
\multirow{2}{*}{Graph} & \multicolumn{3}{c}{Core ball} & \multicolumn{3}{c}{BFS ball} \\
\cline{2-4} \cline{5-7}
 & $\phi_{\mu}$ & $|C|$ & t(s) & $\phi_{\mu}$ & $|C|$ & t(s) \\
\hline
DBLP & 0.011 & 1397 & 207 & 0.029 & 375 & 183 \\
Contact-primary-school & 0.622 & 64 & 154 & 0.669 & 95 & 489 \\
Enron & 0.161 & 154 & 20 & 0.200 & 584 & 30 \\
Contact-high-School & 0.090 & 129 & 21 & 0.090 & 112 & 66 \\
Conference & 0.701 & 95 & 139 & 0.895 & 155 & 297 \\
Hospital & 0.682 & 21 & 15 & 0.906 & 37 & 27 \\
coauth-MAG-History & 0.134 & 334 & 30 & 0.180 & 182 & 21 \\
Email-Eu & 0.623 & 85 & 237 & 0.846 & 269 & 897 \\
\hline
\textbf{Overall} & 0.378 & 285 & -- & 0.477 & 226 & -- \\
\hline
\end{tabular}
\caption{Performance of Core ball and BFS ball on various datasets. The running time for BFS ball includes the full $\alpha=3$ repetition process.}
\label{tab:performance}
\end{table}

\section{Conclusion}
In summary, we introduced a novel approach to local hypergraph clustering by leveraging higher-order motifs.
Unlike traditional methods that often rely on pairwise relationships, our framework integrates specific structural properties from hypergraphs to optimize motif conductance and to improve clustering quality. Our methodology extends previous motif-based approaches by introducing core-based and BFS-based clustering strategies adapted to hypergraphs.
Our experiments on real-world datasets from various fields show that our approach enhances clustering accuracy and computational efficiency.

A promising direction for future research is the investigation of motif enumeration for higher-order structures beyond order 3.
While this extension offers deeper insights into complex network interactions, it also presents significant computational and theoretical challenges. 
Addressing the combinatorial explosion inherent in enumerating higher-order motifs requires the development of novel techniques that improve efficiency and scalability.
In this regard, parallel and distributed computing approaches could be helpful to optimize motif extraction, enabling the analysis of large-scale hypergraphs with high-order structures.

\end{document}